\documentclass[12pt,a4paper,oneside]{article}

\usepackage[authoryear,round]{natbib}
\usepackage{amssymb}
\usepackage{amsthm}

\usepackage[figuresright]{rotating}

\usepackage[flushleft]{threeparttable} % http://ctan.org/pkg/threeparttable
\usepackage{booktabs,caption}
\usepackage{enumitem}

\usepackage{color}
\usepackage{amsxtra}
\usepackage{amsfonts}
\usepackage{amssymb}
\usepackage{latexsym}
\usepackage{amsmath}
\usepackage{amsthm}
\usepackage[mathscr]{eucal}
\usepackage{epsfig}          

\usepackage{tikz}   
\usepackage{float}   
\usepackage{circuitikz}      
\usepackage{fancyhdr}

\newcommand{\R}{{\mathbb{R}}}

\newcommand{\N}{{\mathbb{N}}}

\newcommand{\diff}{{\,{\mathrm{d}}}}

\newcommand{\p}{{\mathbb{P}}}

\newcommand{\E}{{\mathbb{E}}}

\newcommand{\eps}{{\varepsilon}}
\renewcommand{\L}{{\mathscr{L}}}

\newtheorem{kor}{Corollary}

\newtheorem{satz}{Theorem}

\newtheorem{lem}{Lemma}

\theoremstyle{remark}	%Styleänderung (Text aufrecht, ...)

\begin{document}

\title{Censored lifespans in a double-truncated sample: Maximum likelihood inference for the exponential distribution}
%\author{$\ast$}

\author{Fiete Sieg \& Anne-Marie Toparkus \& Rafael Wei{\ss}bach \\[2mm] \textit{\footnotesize{Chair of Statistics and Econometrics,}}   \\[-2mm]
	\textit{\footnotesize{  Faculty for Economic and Social Sciences,}} \\[-2mm]
	\textit{\footnotesize{University of Rostock}} \\
}
\date{ }
\maketitle

\renewcommand{\baselinestretch}{1.3}\normalsize

\begin{abstract}
The analysis of a truncated sample can be hindered by censoring. Survival information may be lost to follow-up or the birthdate may be missing. The data can still be modeled as a truncated point process and it is close to a Poisson process, in the Hellinger distance, as long as the sample is small relative to the population. We assume an exponential distribution for the lifespan, derive the likelihood and profile out the unobservable sample size.  Identification of the exponential parameter is shown, together with consistency and asymptotic normality of its M-estimator. Even though the estimator sequence is indexed in the sample size, both the point estimator and the standard error are observable. Enterprise lifespans in Germany constitute our example.        
 \\[2mm]
\end{abstract}

\section{Introduction and sampling design}
\label{sec1}
%% Labels are used to cross-reference an item using \ref command.
In a simple random sample, each individual of the population has the same and parameter-independent selection probability. A left-truncated panel causes two probabilities, individuals who had died before the study have a probability of zero and those remaining have a parameter-dependent probability. The design has prominent applications in economics \citep{heckman1976} and medicine \citep{kalblawl1989} and the statistical analysis has a long history \citep{turnbull1976,woodroofe1985,And0,keiding1990,stute1993b,he2003,gross1996,eriksson2015,weissbachm2021effect}. Panel data are double-truncated when an individual who dies after the study also has a probability of zero \citep{lynden1971,efron1999,shen2010,moreira2010b,shen2014a,moreira2016,emura2017,franchae2019,emura2020,weiswied2021,weisdoer2022}. The European Commission has published business activity since 2018 with a new definition of an enterprise. Together with the national statistical offices, enterprise foundations and closures are reported, although the mere status of enterprise existence is not reported and the data is hence doubly truncated.

The additional obstacle caused by censoring is that the lifespans of observable individuals are only partially available. The lifespan of an enterprise that was founded in 2018 and is still active when a study ends, here after 2019, is right-censored. In the case of a foundation before 2018, the foundation date of the enterprise is not ascertained retrospectively and the lifespan is hence also censored. Table \ref{tabelle_amtsdaten} summarizes the situation in annual aggregates. Models that combine truncation with censoring have been studied but inefficiently, i.e. without maximizing the likelihood \citep{lai1991,honore2012}. 
\begin{table*}[t]
	\caption{Counts of enterprise foundations and enterprise closures 2018 and 2019: Source: RDC of the Federal Statistical Office and Statistical Offices of the Federal States of Germany, AFiD-Panel (UDE), survey years [2018-2019], own calculations$^{\dagger}$}	\label{tabelle_amtsdaten}
	\centering
	\begin{threeparttable}
		\begin{tabular}{llrrcc}
			\hline \hline
			~ & ~ & 2018 & 2019  & $l^{obs}_j$ & $r^{obs}_j$ \\ \hline \hline
			no. of closures  & founded in 2018 & $\times$ & 50,432   & 0 & 0 \\
			& & & \footnotesize{($y^{obs}_j=0.5$)} & & \\
			~ & founded before 2018 & 246,004  & 315,320  & 1 & 0  \\ 
						& & \footnotesize{($y^{obs}_j=0.5$)} & \footnotesize{($y^{obs}_j=1.5$)} & & \\\hline
			no. of foundations & ~ & 219,417  & 248,020  & 0 & 1 \\ 
						& & \footnotesize{($y^{obs}_j=1.5$)}\tnote{$\ddagger$} & \footnotesize{($y^{obs}_j=0.5$)} & & \\
			\hline \hline
		\end{tabular}
		\begin{tablenotes}
			\item[$\dagger$] Symbols $y^{obs},l^{obs},r^{obs},$ and $j$ are introduced in Section \ref{mal}.
			\item[$\ddagger$] Code excludes the 50,432 uncensored enterprises of the first row.
		\end{tablenotes}
	\end{threeparttable}
\end{table*}

%%%%%%%%%%%%%%%%%%%%%%%%%%%%%%%%%%%%%%%%%%%%%%%%%%%%%%%%%%%%%%%%%%%%%%%%%%%%%%%%%%%%%%%%%%%%%%%%%%%%%%%%%%%%%%%%%
%\newpage

\section{Model, assumptions and likelihood} \label{mal}

We assume a panel study that registers the birth and the death event of its units. The study starts at some point and continues for $s$ time units (namely years). We define the population as units born $G-s$ years before the study beginning until the end of the study. Denote a simple random sample of units $i=1,\dots, n$ that are drawn from the population with measurements $X_i$ and $T_i$ and the probability space $(\Omega, \mathcal F, \{\p_{\theta}:\theta \in \Theta\})$. Here $X_i$ is the lifespan and $T_i$ the birthdate, reformulated as the age of the unit at the beginning of the study. For our example to be worked out in Section \ref{secexa}, Table \ref{tabelle_amtsdaten} lists the realized measurements, aggregated to a discrete-time scale. We elaborate the model for the time-continuous scale.  
\begin{enumerate}[label= \textnormal{(A\arabic{*})}, leftmargin=1cm]
	\item\label{A1:Fiete} Let $\Theta = [\eps, 1/\eps]$ for a small $\eps \in (0,1)$.
	\item\label{A2:Fiete} Let $X_i \sim \text{Exp}(\theta_0)$ with $\theta_0 \in (\eps, 1/\eps)$ and $T_i \sim \text{Unif}([-s, G-s])$ for fixed $0<s<G$.
	\item\label{A3:Fiete} Let $X_i$ and $T_i$ be statistically independent.
\end{enumerate}
By the design of the study, a sample unit is only observable subject to additional conditions. We follow \cite{honore2012} and set $Y_i := X_i - T_i$ for the case of a unit born before the study, i.e. for $T_i > 0$, $Y_i :=X_i$ for the case of a unit with birth and death during the study, i.e. for $T_i \leq 0$ and $X_i \leq T_i + s$, and finally $Y_i:=T_i + s$  for the case of a birth during the study but with death after the study, i.e. for $T_i \leq 0$ and $T_i + s < X_i$. Furthermore, we define indicators for truncation and censoring $L_i  := \chi_{\{T_i > 0\}}$ and $R_i := \chi_{\{T_i + s < X_i \}}$ and describe the data without left- and right-truncated lifespans. 
\begin{enumerate}[resume*]
	\item\label{A4:Fiete} Let $(Y_i, L_i, R_i)$ be unobserved if (i) $(L_i,R_i)=(1,1)$ or (ii) if $(L_i,R_i)=(1,0)$ and $Y_i < 0$.
\end{enumerate}
One can show that restricting use to the reformulated measurements does not result in a loss of information, that is, the likelihoods are equal. For $D := [0,s] \times  \{(0,0),(0,1),(1,0) \}$  and  $\theta \in \Theta$, the probability of unit $i$ to be observed is 
\begin{equation}\label{alpha} \alpha_\theta := \p_\theta \big((Y_i, L_i, R_i) \in D \big)= \frac{s}{G} + \frac{1}{G\theta} \big(1-e^{-\theta s}\big) \big(1-e^{-\theta(G-s)}\big) > 0.
\end{equation}
Closed-form expressions of first and second derivatives with respect to $\theta$, $\dot{\alpha}_{\theta}$ and $\ddot{\alpha}_{\theta}$, are easily derived. We denote the dataset of observations as $\{(y^{obs}_j, l^{obs}_j, r^{obs}_j)\}_{j\leq m}$, where $m$ is the realized random number of observations $M$. In the early literature on truncation, \cite{heckman1976} for instance did not use different indices for latent and observed units and their measurements, but indicated the observed $m$ units as sorted to the beginning of dataset. \cite{shen2014a} considers the situation where the measurements of the unobserved sample units are unknown, but at least the number of unobserved units $n-m$ is known. Of course, an observation $(Y^{obs}_j, L^{obs}_j, R^{obs}_j)$ is not measurable with respect to $\p_{\theta}$. An argument from \cite{weiswied2021} suggests for large $n$ that, under the Assumptions \ref{A1:Fiete}-\ref{A4:Fiete} and for $(y, l, r) \in D$
\begin{equation}\label{ylrdichte}
	f_{\theta}(y, l, r)  =   \frac{1}{\alpha_{\theta} G} e^{-\theta y} \big \{1-e^{-\theta (G-s)} \big \}^l \big \{\theta (s-y)\big \}^{(1-l)(1-r)}
\end{equation}
as the density of $(Y^{obs}_j,L^{obs}_j,R^{obs}_j)$ with respect to the measure ${\p}^{obs}_{\theta}$ that corresponds to the population restricted by $D$. With the reformulation of the data generation starting from $(Y_i, L_i, R_i)$ instead of $(X_i, T_i)$, it is straightforward to extend the likelihood approximation for double-truncated durations \citep{weiswied2021} to an approximate likelihood $L(Data \vert \theta, n)$ for double-truncated and censored durations:
\begin{equation}\label{likehood}
	\left (\frac{n}{G}\right)^M e^{3s-n\alpha_{\theta}} \,  e^{-\theta \sum_{j=1}^M Y^{obs}_j} \prod_{j=1}^M\big \{1-e^{-\theta (G-s)} \big \}^{L^{obs}_j} \big \{\theta (s-Y^{obs}_j)\big \}^{(1-L^{obs}_j)(1-R^{obs}_j)}
\end{equation}

%%%%%%%%%%%%%%%%%%%%%%%%%%%%%%%%%%%%%%%%%%%%%%%%%%%%%%%%%%%%%%%%%%%%%%%%%%%%%%%%%%%%%%%%%%%%%%%%%%%%%%%%%%%%%%%%%
%\newpage
\section{Identification and estimation}

The elementary definition of identification in the population assumes the data to be a simple random sample. However, following \cite{And0}, Assumptions \ref{A1:Fiete}-\ref{A4:Fiete} formulate the data $\{(y^{obs}_j, l^{obs}_j, r^{obs}_j)\}_{j\leq m}$ to be a truncated sample. \cite{efron1999} instead assume the data to be a simple random sample from the population subject to the restriction $D$. Sampling and truncation are not commutative \cite[see][Figure 2]{topaweis2024}. In order to prove consistency when minimizing \eqref{likehood}, \citet[][Theorem 5.7]{vaart1998} formulates a generalized  definition for identification. A preliminary step in order to prove the generalized definition is to prove identification of $\theta$ in the ordinary sense while assuming that $\{(y^{obs}_j, l^{obs}_j, r^{obs}_j)\}_{j\leq m}$ is a simple random sample from the $D$-restricted population.   
\begin{lem} \label{lemindettrunc}
	The parameter $\theta$ is identified in the $D$-restricted population.
\end{lem}

\begin{proof}
	The probability measure in the restricted population is ${\p}^{obs}_{\theta}=\p_{\theta}/\alpha_{\theta}$. We need to show that from ${\p}^{obs}_{\theta_1}={\p}^{obs}_{\theta_2}$ with $\theta_1, \theta_2\in \Theta$ immediately follows $\theta_1 = \theta_2$. Let be ${\p}^{obs}_{\theta_1}={\p}^{obs}_{\theta_2}$, so that for arbitrary $(y,l,r)$ in the support $D$ it follows for the observable data that ${\p}^{obs}_{\theta_1}( Y^{obs}_j \leq y, L^{obs}_j = l, R^{obs}_j = r ) =  {\p}^{obs}_{\theta_2}(  Y^{obs}_j \leq y, L^{obs}_j = l,  R^{obs}_j = r )$, which means $F_{\theta_1} = F_{\theta_2}$.  Hence it is by \eqref{ylrdichte} 
	\[1= \frac{f_{\theta_1}(y, l, r )}{f_{\theta_2}(y, l, r )} = \frac{\alpha_{\theta_2}}{\alpha_{\theta_1}}  e^{-y(\theta_1 - \theta_2)} \left \{ \frac{1-e^{-\theta_1 (G-s)}}{1-e^{-\theta_2 (G-s)}} \right \}^{l} \left \{ \frac{\theta_1}{\theta_2} \right\}^{(1-l)(1-r)}\]
	on the support of $f_{\theta}$. It now suffices to re-sort the $y$-dependent part on the left and differentiate once with respect to $y$, to result in $(\theta_1 - \theta_2)e^{y(\theta_1 - \theta_2)} = 0$ which is equivalent to $\theta_1 = \theta_2$.
\end{proof}

Maximization of likelihood \eqref{likehood} yields $n = m/\alpha_\theta$ and we profile out $n$. Define $m_\theta(y,l,r) :=\chi_D(y,l,r) \log f_\theta(y,l,r) - \chi_D(y,l,r) C(y,l,r)$,  with $\theta$-independent 
$C(y,l,r) := - \log (G) +  \log(s-y) (1-l)(1-r)$. Then \eqref{likehood} can be represented as  
\begin{equation} \label{crit} 
	\log L (data \vert  \theta, n)\big\vert_{n=m/\alpha_\theta} = \sum_{j=1}^m m_\theta(y^{obs}_j,l^{obs}_j,r^{obs}_j) + C_m'
\end{equation}
with parameter-independent $C'_m:= m \log (m) - m -m \log(G)+3s + \sum_{j=1}^m (1-l_j)(1-r^{obs}_j)\log (s-y^{obs}_j)$.  Define further $M_n(\theta):=\frac{1}{n} \sum_{i=1}^n m_\theta(Y_i, L_i, R_i)$ and $ M(\theta):=\E_{\theta_0} M_n(\theta)$ (where $\E_{\theta}$ is the expectation with respect to $\p_{\theta}$), 
so that maximizing arguments of \eqref{crit} and $M_n(\theta)$ are identical. The first does not need the unobserved $n$ and can be used for the computation, whereas the second is an average and can be used for the analysis.
Under the Assumptions \ref{A1:Fiete}-\ref{A4:Fiete}, short calculations yield the following properties (and especially use Lemma \ref{lemindettrunc} for (e)):
\begin{enumerate}
	\item[(a)] For $0< s < G$, the function $\theta \mapsto m_\theta(y,l,r)$ is continuous on $\Theta$ for all $(y,l,r)\in D$.
	\item[(b)] The function $\theta \mapsto m_\theta$ is dominated by a $\p_{\theta_0}$-integrable $\theta$-independent function.
	\item[(c)] It is $\sup_{\theta\in \Theta} | M_n(\theta) - M(\theta) | \overset{p}{\to} 0$.
	\item[(d)] The function $\theta \mapsto M(\theta)$ is continuous on $\Theta$.
	\item[(e)] The true parameter $\theta_0$ is the unique maximizer of $M(\theta)$.
\end{enumerate}

As maximizing argument of the criterion function $M_n(\theta)$ define
\begin{equation} \label{defschaetz}
	\hat \theta_n := \inf \big \{\hat \theta \in \Theta \, | \, \forall \theta \in \Theta: \,M_n(\hat \theta) \geq M_n(\theta) \big\}.
\end{equation}

\begin{satz}\label{satzcon}
	Under Assumptions \ref{A1:Fiete}-\ref{A4:Fiete} and $0 < s < G$ it is $\hat \theta_n \overset{p}{\to} \theta_0$ for $n \to \infty$.
\end{satz}
\begin{proof}
	The first condition of \citet[][Theorem 5.7]{vaart1998} is $(c)$. The second condition is the generalized identification definition. Note that $\theta_0$ uniquely maximizes $M(\theta)$ due to $(e)$, and that the continuity of $M(\theta)$ results from $(d)$ together with the compactness of $\Theta$. Hence, each sequence $(\tilde \theta_n)_{n\in \N}$ with $M_n(\tilde \theta_n) \geq M_n(\theta_0)-o_P(1)$ converges in probability against $\theta_0$. The latter also holds for the sequence defined in \eqref{defschaetz} because $\hat \theta_n$ is the global maximizer of $M_n$ and hence $M_n(\hat \theta_n) \geq M_n(\theta_0)$.
\end{proof}

The consistency is one condition for the asymptotic normality of $\hat \theta_n$ and we now state the remaining conditions.
\begin{kor}\label{korro1}
	It is $\theta \mapsto m_{\theta}(y,l,r)$ two times differentiable in $\theta$ on $(\eps, 1/ \eps)$ for all $(y,l,r)\in D$.
\end{kor}
\begin{kor}\label{korro2}
	It is $ \frac{d^2}{d \theta^2} M(\theta) = \E_{\theta_0}\, \frac{d^2}{d \theta^2} m_{\theta}(Y_i, L_i, R_i)$.
\end{kor}
For all $(y,l,r)\in D$ and $\theta_1, \theta_2 \in \Theta$ it is 
\begin{equation} \label{lemma9}		|m_{\theta_1}(y,l,r) - m_{\theta_2}(y,l,r)|\leq \dot{m}(y,l,r) |\theta_1 - \theta_2|
\end{equation}
for a measurable bound $\dot{m}$ with $\E_{\theta_0} \dot{m}(Y_i, L_i, R_i)^2< \infty$.

As final preparation, and similar to the Fisher information in maximum likelihood theory, we have (with proof in \ref{app1}):
\begin{lem}\label{lemmawichtig}
	For $0< s < G$ and for any $\theta_0 \in \Theta$  it is $\frac{d^2}{d \theta^2} M(\theta_0)\neq 0$.
\end{lem}
Here we write as usual for a function $g$, $\frac{d}{dx}g(a)$ short for $\frac{d}{dx}g(x)\vert_{x=a}$.
\begin{satz}\label{normal}
	Under the Assumptions \ref{A1:Fiete}-\ref{A4:Fiete} and $0 < s < G$, the sequence $ \{\sqrt{n} (\hat \theta_n - \theta_0)  \}_{n\in \N}$ is asymptotically normally distributed with expectation $0$ and variance 
	$\sigma^2 := \E_{\theta_0}[\frac{d}{d \theta} m_{\theta_0}(Y_i, L_i, R_i)^2  ]  /   \E_{\theta_0} [ \frac{d^2}{d \theta^2} m_{\theta_0}(Y_i, L_i, R_i)  ]^2$.
\end{satz}

\begin{proof}
	The proof applies Theorem 5.23 of \cite{vaart1998}. The mapping $\theta \mapsto m_\theta(y,l,r)$ is differentiable for all $(y,l,r)\in D$ and the bound \eqref{lemma9} holds. Furthermore, $\theta \mapsto M(\theta)$ allows a second-order Taylor expansion in the maximizer $\theta_0 \in (\eps, 1/\eps)$ by Corollaries \ref{korro1} and \ref{korro2}. The nonsingularity given by Lemma \ref{lemmawichtig} is an additional condition.  The requirement $M_n(\hat \theta_n) \geq \sup_{\theta \in \Theta} M_n(\theta) - o_P(n^{-1})$ holds by virtue of the definition of $\hat \theta_n$. Consistency of the estimator holds by virtue of Theorem \ref{satzcon}.
\end{proof}

%%%%%%%%%%%%%%%%%%%%%%%%%%%%%%%%%%%%%%%%%%%%%%%%%%%%%%%%%%%%%%%%%%%%%%%%%%%%%%%%%%%%%%%%%%%%%%%%%%%%%%%%%%%%%%%%%
%\newpage
\section{Data example} \label{secexa}

We now assume the exponential distribution as a model for the lifespans of German enterprises $X$ (Assumption \ref{A2:Fiete}). In order to ensure data privacy, the exact days of foundation and closure, and hence $x^{obs}_j$ and $t^{obs}_j$, for the annual data in Table \ref{tabelle_amtsdaten}, are not available. In order to demonstrate the method, we assume that any event has occurred in the middle of a year. For instance, the 50,432 enterprises with observed foundation and closure years ($l_j=0,r_j=0$) survived $x^{obs}_j=y^{obs}_j=1$ year. Because the study begin is not an event, $t^{obs}_j$ cannot be defined by $x^{obs}_j-y^{obs}_j$. For a genuinely time-discrete model, $t^{obs}_j$ is the age one year before the study begin \cite[see][for a different design]{scholzweis2025}. For simplicity we assume the begin of the study between 2017 and 2018, so that the 246,004 enterprises with left-truncated  lifespan ($l_j=1$) have been  observable under risk for half a year, i.e. $y^{obs}_j=0.5$. The remaining counts in  Table \ref{tabelle_amtsdaten} have similar explanations and especially Assumption \ref{A4:Fiete} is fulfilled for any number of foundation cohorts $G > 2=s$. Estimator  \eqref{defschaetz} requires maximizing $(n/m) M_n (\theta)  = m^{-1} \sum_{i=1}^n  \chi_D(y_i, l_i, r_i)  \{ - \log \alpha_{\theta}  -\theta  y_i +K(\theta)  l_i + \log(\theta)  (1- l_i)(1-r_i) \} 
=  m^{-1} \sum_{j=1}^m   \{ - \log \alpha_{\theta}  -\theta   y^{obs}_j +K(\theta)  l^{obs}_j + \log(\theta) (1-  l^{obs}_j)(1- r^{obs}_j) \} 
=  - \log \alpha_{\theta} - \theta\cdot 0.9952764 + \log  ( 1-e^{-\theta (G-2)} ) \cdot 0.5456311+ \log(\theta) \cdot 0.0490221$, with $\alpha_{\theta}$ given in \eqref{alpha}.
The verifiable assumptions of Theorem \ref{normal} are fulfilled, and for the calculation of the standard error of $\hat{\theta}_n$, the expectations in $\sigma^2$ of Theorem \ref{normal} can be replaced by averages:
\begin{equation} \label{SE}
	\frac{\hat \sigma^2}{n} := \frac{1}{n}  \frac{\frac{1}{n} \sum_{i=1}^n \frac{d}{d \theta} m_{\hat \theta_n}(y_i, l_i, r_i)^2 }{\frac{1}{n^2}[ \sum_{i=1}^n \frac{d^2}{d \theta^2} m_{\hat \theta_n}(y_i, l_i, r_i) ]^2}
\end{equation}
With the formulae for $\dot{\alpha}_{\theta_0}$ and $\ddot{\alpha}_{\theta_0}$, the standard error is observable, (i) because the indicator reduces the involved sums to $m$ observations, (ii) because $n$ cancels out and (iii) because $\theta_0$ is replaced by $\hat \theta_n$. Numerical results for some $G$ are listed in Table \ref{tabelle_G} and suggest that $\hat \theta_n$ depends on the considered foundation cohorts of $G$ ending with 2019. 
\begin{table}[h]
	\centering
	\caption{Results for data from Table \ref{tabelle_amtsdaten} for different populations: Point estimates \eqref{defschaetz}, selection probability \eqref{alpha}, standard error \eqref{SE}}	\label{tabelle_G}
	\begin{tabular}{rcrcc}
		\hline \hline 
		$G$ & $\hat\theta_n$ & Life expectancy & $\alpha_{\hat \theta_n} $  & Standard error $\hat\sigma / \sqrt{n} \; (\cdot 10^{-4})$  \\ \hline
		5 & 0.2818 &   3.55 & 0.574 &  3.03  \\ 
		10 & 0.1849 &   5.41 & 0.329 & 2.48  \\ 
		15 & 0.1492 &   6.70 & 0.232 &  2.36  \\  
		30 & 0.1111 &   9.00 & 0.124 & 2.58  \\ 
		50 & 0.0972 & 10.28 & 0.076 &  3.13  \\ 
		100 & 0.0922 & 10.85 & 0.038 &  3.78  \\ 
		200 & 0.0921 & 10.86 & 0.019 &  3.82  \\ \hline \hline
	\end{tabular} 
\end{table}
In practical terms, one must keep in mind that for an extensively small $G$, the data may contain enterprises founded before year $2020-G$, i.e. outside the population. Furthermore if $\alpha_{\theta_0}$ becomes too large, the truncated point process can no longer be approximated well by a Poisson process. Criticism may arise from the simplicity of assuming the closure hazard to be constant across ages (Assumption \ref{A2:Fiete}). That the hazard is constant in time, and equivalently lifespan and foundation date are independent (Assumption \ref{A3:Fiete}), may also be inadequate \citep[see][for an older dataset and without censoring]{topaweis2024}. Also, some indication against a constant intensity of foundations (Assumption \ref{A2:Fiete}) also already exists \citep{weisdoer2022}.

%%%%%%%%%%%%%%%%%%%%%%%%%%%%%%%%%%%%%%%%%%%%%%%%%%%%%%%%%%%%%%%%%%%%%%%%%%%%%%%%%%%%%%%%%%%%%%%%%%%%%%%%%%%%%%%%%
%%%%%%%%%%%%%%%%%%%%%%%%%%%%%%%%%%%%%%%%%%%%%%%%%%%%%%%%%%%%%%%%%%%%%%%%%%%%%%%%%%%%%%%%%%%%%%%%%%%%%%%%%%%%%%%%%

%% The Appendices part is started with the command \appendix;
%% appendix sections are then done as normal sections
%\newpage
\appendix

\section{Nonsingularity of the Fisher information  (Lemma \ref{lemmawichtig})}
\label{app1}

Now write $D_\theta^k$ short for $d^k/d \theta^k$ ($k=1,2$). According to the proof of \eqref{lemma9}, $D^1_\theta m_\theta(y,l,r)$ is bounded and expectation and differentiation can be interchanged in the following equation:  
\[D^1_\theta M(\theta) = D^1_\theta \E_{\theta_0} m_\theta (Y_i, L_i, R_i) = \E_{\theta_0}  D^1_\theta m_\theta (Y_i, L_i, R_i) \]

Now 
\begin{eqnarray*}
D_\theta^2M(\theta_0) & = & \E_{\theta_0} [ D_\theta^2 m_{\theta}(Y_i, L_i, R_i)  ]\big |_{\theta = \theta_0} \\
& = & \E_{\theta} [ D_\theta^2 m_{\theta}(Y_i, L_i, R_i)  ]\big |_{\theta = \theta_0}=: \eta(\theta \, | \, s, G)\big|_{\theta = \theta_0}.
\end{eqnarray*} 
Important here is that the expectation with respect to $\p_{\theta_0}$ could be relaxed to an expectation with respect to the measure $\p_\theta$. 
If $ \eta(\theta \, | \, s, G) \neq 0$ for all $\theta \in \Theta$; this will be especially true for $\theta_0$ and will suffice for Lemma \ref{lemmawichtig}.

Note first that (using the uniform distribution of $T_i$):
\begin{equation} \label{elelr}
	\begin{split}
		\E_{\theta} \chi_D(Y_i, L_i, R_i)  L_i = \alpha_{\theta} - \frac{s}{G} \\
		\E_{\theta} \chi_D(Y_i, L_i, R_i) ( 1-L_i)(1-R_i) = \frac{s}{G} - \frac{1}{G \theta} \big (1-e^{-\theta s} \big )
	\end{split}	
\end{equation}

For $0<s<G$ and  $K''(\theta):= -\frac{(G-s)^2}{e^{\theta (G-s)}-1}  -\frac{(G-s)^2}{(e^{\theta (G-s)}-1)^2}$ it is
\begin{equation} \label{darstellungeta}
	\eta(\theta \, | \, s, G) = \Big\{ -\ddot{\alpha}_\theta   + \frac{  \dot{\alpha}_\theta^2 }{\alpha_\theta}  \Big \} + K''(\theta) \left \{ \alpha_{\theta} - \frac{s}{G} \right \} - \frac{1}{\theta^2} \left \{ \frac{s}{G} - \frac{1}{G \theta} \big (1-e^{-\theta s} \big ) \right \}.
\end{equation}

The representation follows directly from \eqref{elelr}, the form of $D_\theta^2 m_\theta(y,l,r)$, the proofs of Corollaries \ref{korro1} and \ref{korro2} and the definition of $\eta$ as $\E_\theta$-expectation thereof. 

The following helps to show the global negativity of $\eta$. The middle summand in \eqref{darstellungeta} can be written as 
\begin{equation} \label{repkstrichstrich}
	\ddot{\alpha}_\theta   - \frac{  \dot{\alpha}_\theta^2 }{\alpha_\theta - s/G}  - \frac{1}{\theta^2} \left \{ \alpha_{\theta} - \frac{s}{G} \right \} + \frac{s^2 e^{-\theta s}}{(1-e^{-\theta s})^2} \left \{ \alpha_{\theta} - \frac{s}{G} \right \},
\end{equation}
because by \eqref{alpha}  and  $K(\theta) := \log (1-e^{-\theta (G-s)})$ it is $\log \{ \alpha_\theta - s/G \}  = \log \{(G \theta)^{-1} (1-e^{-\theta s})(1-e^{-\theta (G-s)} )\}  = - \log (G) - \log (\theta) + \log (1-e^{-\theta s} \big ) + K(\theta)$. Differentiation with respect to $\theta$ twice, results in  
$\ddot{\alpha}_\theta (\alpha_\theta - s/G) -  \dot{\alpha}_\theta^2 /(\alpha_\theta - s/G)^2) = \theta^{-2} - s^2 e^{-\theta s}/(1-e^{-\theta s})^2 + K''(\theta)$. Re-sorting the equality and multiplication by $(\alpha_\theta - s/G)$ yields \eqref{repkstrichstrich}.

For $0<s<G$ we may write
\begin{equation} \label{darstellung2eta}
	\eta(\theta  \vert  s, G)  = \left\{ \frac{  \dot{\alpha}_\theta^2 }{\alpha_\theta} - \frac{  \dot{\alpha}_\theta^2 }{\alpha_\theta - s/G}   \right \}  + \alpha_\theta H_2(\theta)  - \frac{1}{G}  (1-e^{-\theta s} ) H_3(\theta),
\end{equation}
with $H_k(\theta):= s^k   (1-e^{-\theta s})^{-k} e^{-\theta s} - \theta^{-k}$, for $k=2,3$.
The representation is obtained by inserting \eqref{repkstrichstrich} in \eqref{darstellungeta}. Some ratios are eliminated and what remains are
\begin{multline*}
	\eta(\theta \vert  s, G)  =  \frac{  \dot{\alpha}_\theta^2 }{\alpha_\theta} - \frac{  \dot{\alpha}_\theta^2 }{\alpha_\theta - s/G}   - \frac{1}{\theta^2} \alpha_\theta + \frac{s^2 e^{-\theta s}}{(1-e^{-\theta s})^2} \left \{ \alpha_{\theta} - \frac{s}{G} \right \} \\ + \frac{1}{G\theta^3}  (1- e^{-\theta s}  )  
	=  \left\{ \frac{  \dot{\alpha}_\theta^2 }{\alpha_\theta} - \frac{ \dot{\alpha}_\theta^2 }{\alpha_\theta - s/G}   \right \} + \alpha_\theta  \left\{ \frac{s^2 e^{-\theta s}}{(1-e^{-\theta s})^2} - \frac{1}{\theta^2}  \right \} \\ - \frac{1}{G} (1-e^{-\theta s}) \left\{ \frac{s^3 e^{-\theta s}}{(1-e^{-\theta s})^3} - \frac{1}{\theta^3}  \right\}
\end{multline*}
and the $H_k(\theta)$ are suitably defined for \eqref{darstellung2eta}. 

Now note that for the first summand (in brackets) in \eqref{darstellung2eta}, it holds $\dot{\alpha}_\theta^2/\alpha_\theta -  \dot{\alpha}_\theta^2 /(\alpha_\theta - s/G) \leq 0$  	
for all $\theta \in \Theta$ and $0<s< G$. The latter follows from the fact that obviously  $\dot{\alpha}_\theta^2 \geq 0$ and by \eqref{alpha} holds across $\Theta = [\eps, 1/ \eps]$ that 
$0 < \alpha_\theta - s/G< \alpha_\theta  \; \Rightarrow \; \alpha_\theta^{-1} < (\alpha_\theta - s/G)^{-1} \; \Rightarrow \;  \alpha_\theta^{-1} - (\alpha_\theta - s/G)^{-1} < 0$.

Note finally that summing the second and third summands in \eqref{darstellung2eta}, for $0< s <G$, $H_1(\theta):= s e^{-\theta s} /(1-e^{-\theta s})  - 1/ \theta$ and all $\theta \in \Theta$ yields after a short calculation:
\begin{equation} \label{lemdelta}
	\Delta:=\underbrace{ \frac{s}{G \theta}  (1-e^{-\theta (G-s)}  ) }_{\text{$>0$}} \underbrace{H_1(\theta)}_{\text{$<0$}} - \underbrace{\frac{s}{G \theta^2} (1-e^{-\theta s}  ) e^{-\theta(G-s)}}_{\text{$>0$}} \underbrace{ \left \{ \frac{1}{s} H_1(\theta) + 1 \right \}}_{\text{$>0$}} <0 
\end{equation}

\qed

\vspace*{0.7cm}
\noindent \textbf{Declarations}:
The authors declare that they have no conflict of interest.

\vspace*{0.3cm}
\textbf{Acknowledgment}:
We thank Dr. Claudia Gregor-Lawrenz (BaFin, Bonn) for continuous support on the topic, e.g. by contributing earlier data in the LTRC-design. We thank Dr. Florian K\"ohler (Destatis, Hannover) for supply of the data, Dr. Wolfram Lohse (Görg, Hamburg) for support in the data acquisition process, Simon Rommelspacher (Destatis, Wiesbaden) for advise on the measurement definitions. The financial support from the Deutsche Forschungsgemeinschaft (DFG) of R. Wei\ss bach is gratefully acknowledged (Grant 386913674 ``Multi-state, multi-time, multi-level analysis of health-related demographic events: Statistical aspects and applications'').

\section{Supplement}

If in a calculation only one individual is concern, the index will be dropped, $(X, T)$ instead of $(X_i, T_i)$, $(Y, L, R)$ instead of $(Y_i, L_i, R_i)$, $(X^{obs},T^{obs})$ instead of $(X^{obs}_j,T^{obs}_j)$, $(Y^{obs}, L^{obs}, R^{obs})$ instead of $(Y^{obs}_j, L^{obs}_j, R^{obs}_j)$, $(Y^*, L^*, R^*)$ instead of $(Y_l^*, L_l^*, R_l^*)$, $(Y^0, L^0, R^0)$ instead of $(Y_i^0, L_i^0, R_i^0)$.

\subsection{Derivation of observation probability $\alpha_\theta$ (Formula \ref{alpha})}

A short calculation yields that it is $(X, T) \in D_0$ for $D_0 :=\big\{(x,t)\big| 0 < t \leq x \leq t+s \leq G\big\} \cup \big\{(x,t) \big| -s \leq t \leq 0\big\}$ if and only if  $(Y, L, R)\, \in \, D$. Therefore $\alpha_\theta   = \p_\theta \big((Y, L, R) \in D \big) =  \p_\theta \big((X, T) \in D_0 \big)$. The set $D_0$ is the function of two disjoint sets so that:
\begin{align*}
	\alpha_\theta & = \p_\theta \big ( -s\leq T \leq 0 \big ) + \p_\theta \big( 0 < T \leq X \leq T+s \leq G  \big )  =: P_1 + P_2 
\end{align*}
According to Assumptions \ref{A2:Fiete}, obviously $P_1 = s/G$. The property of the conditional expectations yields (with $\E$ as expectation with respect to $\p$)
\begin{align*}
	P_2 & =  \E_\theta  \big( \chi_{(0, G-s]}(T) \, \chi_{[T , T + s]}(X)  \big ) 
	= \E_\theta \big ( \E_\theta \big [ \chi_{(0, G-s]}(T) \, \chi_{[T , T + s]}(X) \, \big | \, T \big ] \big ) \\
	& = \E_\theta \big ( \chi_{(0, G-s]}(T) \, \E_\theta \big [ \chi_{[T , T + s]}(X) \, \big | \, T \big ] \big ) \\
	& = \E_\theta \big ( \chi_{(0, G-s]}(T) \, \p_\theta \big [ T \leq X \leq T + s \, \big | \, T  \big ] \big ) .
\end{align*}
In the third equality, note that $ \chi_{(0, G-s]}(T)$ is measurable with respect to $\sigma(T)$. Again using Assumptions \ref{A2:Fiete} it is $P_2 = \E_\theta ( \chi_{(0, G-s]}(T)  \{ 1-\exp(-\theta (T +s)) - 1+\exp(-\theta T )\}  ) = G^{-1} \int_{0}^{G-s}  \{ e^{-\theta u} - e^{-\theta (u+s)} \}  \diff u 
= G^{-1}  (1-e^{-\theta s}  ) \int_{0}^{G-s} e^{-\theta u}  \diff u   = G^{-1} (1-e^{-\theta s} ) \frac{1- e^{-\theta (G-s)}}{\theta}$. 
\qed

%\newpage 

%%%%%%%%%%%%%%%%%%%%%%%%%%%%%%%%%%%%%%%%%%%%%%%%%%%%%%%%%%%%%%%%%%%%%%%%%%%%%%%%%%%%%%%%%%%%%%%%%%%%%%%%%%%%%%%%%

\subsection{Derivation of observed PDF $f_{\theta}(y, l, r)$ (Formula \ref{ylrdichte})}

The data is a truncated point process and equivalent to a mixed empirical process with binomially distributed number of summands. By approximating the latter by a Poisson distributed random variable a Poisson process approximates the data and the distribution of its atoms $(X^{obs}_j,T^{obs}_j)$ is the distribution of $(X_i, T_i)$ conditional on observability, i.e. on $T_i \le X_i \le T_i +s$ \cite[see][]{weiswied2021}.   
Written in terms $( Y_i, L_i, R_i)$ and $(Y^{obs}_j, L^{obs}_j, R^{obs}_j)$ for $B \in \mathcal B$ it is
\begin{equation*} \p_{\theta_0}\big ( ( Y, L, R) \in B \, \big | \, ( Y, L, R) \in D  \big )
	=\p_{\theta_0}\big ( ( Y, L, R) \in B \cap D \big ) \, \big / \,  \p_{\theta_0}\big ( ( Y, L, R) \in D  \big ) 
\end{equation*}
and it follows 
\begin{eqnarray*}
	F_{\theta_0}(y,l,r)  & = & \p_{\theta_0} \big ( (Y , L, R) \in [-(G-s) , y]\times (l,r) \cap D\big ) \, \big / \, \p_{\theta_0} \big (Y , L, R) \in D \big ) \\
	& = &  \p_{\theta_0} \big ( Y \in  [-(G-s) , y] \cap [0,s], 
	(L, R) \in  \{(l,r)\} \\
	& &  \cap \{(0,0),(0,1),(1,0)\} \big ) \, \big / \, \alpha_{\theta_0}\\
	& = & \p_{\theta_0}\big( Y\in  [-(G-s) , y] \cap [0,s], (L , R) \in \{(l,r)\} \cap D^o \big ) \, \big / \, \alpha_{\theta_0},
\end{eqnarray*}
where $D^o:=\{(0,0),(0,1),(1,0)\}$. Hence, the distribution of observation $\{(Y^{obs}_j, L^{obs}_j, R^{obs}_j)\}_{j\leq m}$ on its support $S$ is given 	for $(y,l,r)\in S:= (0,s] \times \{(0,0) \} \cup [0,s] \times \{(0,1) \} \cup (-(G-s),s] \times  \{(1,0)  \} \cup (s, \infty) \times \{(1,1)\}$ by
\begin{equation}	 \label{Lemma3}
	F_{\theta_0}(y, l, r)  
	= \frac{1}{\alpha_{\theta_0}} \,\p_{\theta_0}\big( Y \in [-(G-s), y]\cap [0,s],\, (L , R) \in \{(l,r)\} \cap D^o \big).
\end{equation}
From the distribution of $(Y^{obs}_j, L^{obs}_j, R^{obs}_j)$ its density (with respect to the product measure of Lebesgue measure and two count measures) derives. 	The four disjoint outcomes  $\{(L, R)=(i,j)\}$ for $i,j \in \{0,1\}$ will now be considered separately. For $y < 0$ and  $y>s$, $F_{\theta_0}$ remains on constant level to that related parts vanish after differentiation with respect to $y$. Hence let be without loss of generality $y\in [0,s]$. Due to \eqref{Lemma3} it is
\begin{align*}
	F_{\theta_0}(y,0,0) & =  \frac{1}{\alpha_{\theta_0}} \, \p_{\theta_0} \big (0\leq Y \leq y,\, (L , R) = (0,0) \big ) 
\end{align*}
and according to the definitions of $Y, L$ and $R$ the above probability results from the distribution of $X$ and $T$ as follows. 
\begin{align*}
	\alpha_{\theta_0} F_{\theta_0}(y,0,0) & = \p_{\theta_0} \big ( X \leq y, \, T \leq 0, \, X \leq T + s\big ) \\
	& =  \E_{\theta_0}  \big( \chi_{[-s, 0]}(T) \, \chi_{[0 ,\, \min(y, \, T + s)]}(X)  \big ) \\
	& =  \E_{\theta_0}  \big( \E_{\theta_0} \big [ \chi_{[-s, 0]}(T) \, \chi_{[0 ,\, \min(y, \, T + s)]}(X) \, \big | \, T \big] \big )  \\
	& =  \E_{\theta_0}  \big( \chi_{[-s, 0]}(T) \,  \E_{\theta_0} \big [ \chi_{[0 ,\, \min(y, \, T + s)]}(X) \, \big | \, T \big] \big )  
\end{align*}
The last equation uses that $\chi_{[-s, 0]}(T)$ is measurable with respect to $\sigma(T)$ and can hence be written before the inner conditional expectation. By the exponential distribution of $X$ with parameter $\theta_0$ it is for $g(t):=\min(y, t+s)$ almost surely 
$\E_{\theta_0} [ \chi_{[0,g(T)]}(X)  \vert  T ] = F^{\text{Exp}} (g(T)  \vert  \theta_0 ) = 1-e^{-\theta_0 g(T)}$. Together with the uniform distribution of $T$ on $[-s, G-s]$ it is  
\begin{align*}
	\alpha_{\theta_0}  F_{\theta_0}(y,0,0) & =  \E_{\theta_0}  \big( \chi_{[-s, 0]}(T) \big \{1-e^{-\theta_0 \min(y, \, T + s)} \big \}  \big )  \\
	& = \frac{1}{G} \int_{-s}^0 \left \{ 1- e^{-\theta_0\min(y, \, u + s) }\right \} \diff u. 
\end{align*}
Now, note that $\min (y, u+s)=u+s$ holds if and ony if $u \leq -s +y$ and the regions of integration are separated by $y\in [0,s]$ adequately and it is  
\begin{align*}
	\alpha_{\theta_0}  F_{\theta_0}(y,0,0)  & = \frac{s}{G} - \frac{1}{G}\int_{-s}^{-s+y} e^{-\theta_0( u + s)} \diff u - \frac{1}{G} \int_{-s+y}^0  e^{-\theta_0y}\diff u \\ 
	& = \frac{s}{G} - \frac{1}{G}\int_{0}^{y} e^{-\theta_0 v} \diff v + \frac{1}{G}  \int_{0}^{y-s}  \diff u \cdot e^{-\theta_0y} \\ 
	& = \frac{s}{G} -  \frac{1-e^{-\theta_0 y}}{G \theta_0} + \frac{1}{G} (y-s)e^{-\theta_0 y}.
\end{align*}
Differentiation with respect to $y$ yields $f_{\theta_0}(y, 0, 0) = (\alpha_{\theta_0} G)^{-1} \{ -e^{-\theta_0 y} + e^{-\theta_0 y} - \theta_0(y-s) e^{-\theta_0 y} \} 
=  (\alpha_{\theta_0} G)^{-1} e^{-\theta_0 y}  \{ \theta_0(s-y) \}$. 
Similar for the cases $(l,r)=(0,1), (1,0), (1,1)$, it is 
\begin{multline*}
	F_{\theta_0}(y,1,1)  = \p_{\theta_0}  (Y \leq y, L = 1, R = 1  \vert  (Y, L, R) \in D  ) \\ 
	= \p_{\theta_0} (Y \leq \min(y,s), (L , R) = (1,1) ,(L, R) \in \{(0,0),(0,1),(1,0)\} )  /  \alpha_{\theta_0} 
	= 0, 
\end{multline*}
so that immediately $f_{\theta_0}(y,1,1)\equiv 0$. For the remaining two $(l,r)$, for $y\in [0,s]$, again the equality $\alpha_{\theta_0} F_{\theta_0}(y,l,r)  = \p_{\theta_0} \big (Y \leq y, L = l, R = r \big )$ holds.
Therefore we have for $(l,r)=(0,1)$ 
\begin{align*}
	\alpha_{\theta_0}  F_{\theta_0}(y, 0, 1) & = \p_{\theta_0} \big ( T + s \leq y, \, T \leq 0, \, T + s < X \big ) \\
	& = \E_{\theta_0} \big ( \chi_{[-s, -s+y]}(T) \, \chi_{(T + s,\infty)}(X) \big ) \\
	& = \E_{\theta_0} \big ( \chi_{[-s, -s+y]}(T) \, \E_{\theta_0} [ \chi_{(T + s,\infty)}(X) \, | \, T] \big ) \\
	& = \frac{1}{G}\int_{-s}^{-s+y} \big \{1 - 1 + e^{-\theta_0 (u+s)} \big\} \diff u 
	= \frac {1}{G} \int_0^y e^{-\theta_0 v} \diff v 
\end{align*}
and differentiation with respect to $y$ yields $f_{\theta_0}(y,0,1)= (\alpha_{\theta_0}G)^{-1} e^{-\theta_0 y}$.
Note finally that for $(l,r)=(1,0)$ it is  
\begin{align*}
	\alpha_{\theta_0} F_{\theta_0}(y, 1,0) & = \p_{\theta_0} \big ( X - T \leq y, \, T > 0, \, X \leq T + s \big ) \\
	& = \E_{\theta_0} \big ( \chi_{(0, G-s]}(T) \, \chi_{[0, T + y]}(X) \big ) \\
	& = \E_{\theta_0} \big ( \chi_{(0, G-s]}(T) \, \E_{\theta_0} [ \chi_{[0, T + y]}(X) \, | \, T] \big ) \\
	& = \frac{1}{G} \int_{0}^{G-s} \big \{1-e^{-\theta_0(u+y)} \big \} \diff u \\
	& = \frac{G-s}{G} - \frac{1}{G} e^{-\theta_0 y} \int_{0}^{G-s}e^{-\theta_0 u} \diff u \\
	& = \frac{G-s}{G} - \frac{1}{G} e^{-\theta_0 y} \frac{1-e^{-\theta_0 (G-s)}}{\theta_0}
\end{align*} 
so that $f_{\theta_0}(y,1,0)=(\alpha_{\theta_0} G)^{-1} e^{-\theta_0 y} \{ 1-e^{-\theta_0 (G-s)} \}$.
\qed

%\newpage 

\subsection{Derivation of likelihood approximation (Formula \ref{likehood})}

Let denote $\eps_P$ the Dirac-measure in point $P$. The data are now described by the truncated point process 
$N_{n,D} (\cdot ):= \sum_{i=1}^n \eps_{(Y_i, L_i, R_i)}(\cdot \cap D)$, which equals in distribution a Binomial process according \citet[][Theorem 1.4.1]{reiss1993}. It is a mapping $\Omega^n \to \mathbb M(S, \mathcal B)$ with the point measure over $(S, \mathcal B)$ as image, and $\mathcal B$ as $\sigma$-field over $S$. Note that $N_{n,D}(B):\Omega^n \to \{0, \dots , n\}$ with $\omega \mapsto N_{n,D}^\omega(B)$ for each $B\in \mathcal B$ is the random number of $\{(Y_i, L_i, R_i)\}_{1\leq i \leq n}$, that belong to $B \cap D$. Denote the intensity measure by $\nu_{n,D}$ and we will approximate the density of $N_{n,D}$, i.e. the likelihood, by that of a Poisson process $N_n^*$ which has an identical intensity measure. The likelihood approximation can than be maximized in $(n, \theta)$. Both processes are close in Hellinger distance as $\alpha_{\theta_0}$ is small \cite[][Theorem 1.4.2]{reiss1993}, for instance if $s \ll G$.   

Define $Z_n$ to be Poisson-distributed with parameter $n\alpha_{\theta_0}$ and to be independent of a sequence of independent and identically distributed random vectors $\{(Y_l^*, L_l^*, R_l^*)\}_{l\geq 1}$, each with density \eqref{ylrdichte}. Define $N_n^* (\cdot ):= \sum_{l=1}^{Z_n} \eps_{( Y_l^*, L_l^*, R_l^*)}(\cdot)$. 
\begin{lem}\label{lemma3}
	The intensity measure $\nu_n^*$ of $N_n^*$ is equal to  
	$\nu_{n,D}$.
\end{lem}

\begin{proof} For $B\in \mathcal B$ it is $\nu_n^*(B)  = \E_{\theta_0} N_n^*(B) = \E_{\theta_0}  \sum_{l=1}^{Z_n} \eps_{( Y_l^*, L_l^*, R_l^*)}(B)$.
	Due to the independence of $Z_n$, the distribution $Z_n \sim \text{Poi}(n\alpha_{\theta_0})$ and the identical distributions of the $( Y_l^*, L_l^*, R_l^*)$ follows 
	\begin{eqnarray*}
		\nu_n^*(B)& = & \E_{\theta_0} Z_n \, \E_{\theta_0} \eps_{( Y^*, L^*, R^*)}(B) \\
		& = & n\alpha_{\theta_0} \,\p_{\theta_0}\big (( Y^*, L^*, R^*) \in B \big ) =n\alpha_{\theta_0} \, \p_{\theta_0}\big (( Y^{obs}, L^{obs}, R^{obs}) \in B \big ).
	\end{eqnarray*}
	And furthermore $\p_{\theta_0}\big ( ( Y^{obs}, L^{obs}, R^{obs})\in B \big ) = \p_{\theta_0}\big (( Y, L, R) \in B\, | \, ( Y, L, R) \in D\big )$ so that 
	\begin{align*}
		\nu_n^*(B)  & = n\alpha_{\theta_0} \cdot \p_{\theta_0}\big (( Y, L, R) \in B \, \big | \, ( Y, L, R) \in D \big ) \\
		& = n\alpha_{\theta_0} \cdot \p_{\theta_0}\big (( Y, L, R) \in B \cap D \big ) \big / \alpha_{\theta_0} 
		= n \cdot \p_{\theta_0}\big (( Y, L, R) \in B \cap D \big ) .
	\end{align*}
	Finally, the last expression is $\nu_{n,D}(B)$, because $\nu_{n,D}(B)  = \E_{\theta_0} N_{n,D}(B) = \E_{\theta_0} \sum_{i=1}^n \eps_{(Y_i, L_i, R_i)}(B \cap D) = n \cdot \p_{\theta_0}\big (( Y, L, R) \in B \cap D \big )$.
\end{proof}

In order to derive the Radon-Nikodym derivative of $N_n^*$ the dominating measure is chosen to be that of a parameter-independent Poisson process $N_0$. 
To this end, define $(Y_i^0, L_i^0, R_i^0)$ with $Y_i^0 \sim \text{Uni}([0,s])$ and thereof independent $(L_i^0, R_i^0) \sim \text{Uni}( \{(0,0),(0,1),(1,0)\} )$ for $i=1,2,\dots$ as well as $Z_0 \sim \text{Poi}(3s)$, again independent of the both former. Then $N_0 (\cdot) := \sum_{i=1}^{Z_0} \eps_{(Y_i^0, L_i^0, R_i^0)}(\cdot)$ is a 
Poisson process with intensity measure $\nu_0$ and for the sets $B=B_1 \times B_2 \subseteq S$ 
\begin{align*}
	\nu_0 (B) & = \E_{\theta_0} Z_0 \cdot \p_{\theta_0}\big( (Y^0, L^0, R^0) \in B\big )  \\
	& = 3s \cdot \frac{1}{s} \, \lambda_{[0,s]}(B_1) \cdot \frac{1}{3}\,  \text{card}\big(B_2 \cap \{(0,0),(0,1),(1,0)\}\big) \\
	& =  \lambda_{[0,s]}(B_1) \cdot \big \{ \eps_{(0,0)}(B_2) + \eps_{(0,1)}(B_2) +\eps_{(1,0)}(B_2) \big \} =:  \lambda_{[0,s]}(B_1) \cdot \mu_0(B_2),
\end{align*}
where $\lambda_{[0,s]}$ is the Lebesgue measure limited to $[0,s]$. Here $B_1$ denotes an  interval and $B_2$ a subset of $\{0,1\}^2$, so that $B_1 \times B_2 \subset S$. 
\begin{lem}
	The measure $\nu_0$ dominates $\nu_n^*$.
\end{lem}
\begin{proof} We have to show that with $\nu_0 (B) = 0$ also $\nu_n^*(B)=0$. Let be $\nu_0 (B) = 0$ for a $B=B_1\times B_2 \subseteq S$ as above. With the previous result for $\nu_0$ it follows directly $\lambda_{[0,s]}(B_1)=0$ or $B_2 = \{(1,1)\}$. Now, with Lemma \ref{lemma3}  $\p_{\theta_0}\big ((Y^{obs},L^{obs},R^{obs}) \in B \big ) = \sum_{(l,r)\in B_2 \cap \{0,1\}^2} \int_{B_1} f_{\theta_0}(y,l,r) \diff y = 0$ holds now in both cases
	and therefore also $\nu_n^*(B) = n\alpha_{\theta_0} \cdot \p_{\theta_0}\big((Y^{obs},L^{obs},R^{obs}) \in B \big ) =0$.
\end{proof}

The required Radon-Nikodym derivative of $\nu_n^* = \nu_{n,D}$ with respect to $\nu_0$ is a function $h_{\theta_0}:S\to \R_0^+$ so that $\forall B\subseteq S$ holds $\nu_{n,D}(B)= \int_B h_{\theta_0} \diff \nu_0$ \cite[][Theorem 1.2.1.(i)]{reiss1993}.
For $B=B_1\times B_2\subseteq S$ holds as above by Fubini's theorem
\begin{align*}
	\nu_{n,D}(B)& = \int_{B} h_{\theta_0} \diff (\lambda_{[0,s]} \otimes \mu_0) \\
	&	= \sum_{(l,r)\in B_2\cap \{(0,0),(0,1),(1,0)\} } \int_{B_1\cap [0,s]} h_{\theta_0}(y, l,r) \diff y.
\end{align*}
Furthermore, due to Lemma \ref{lemma3} 
\begin{align*}
	\nu_{n,D}(B)  & = n\alpha_{\theta_0} \cdot \sum_{(l,r)\in B_2 \cap \{0,1\}^2} \int_{B_1} f_{\theta_0}(y,l,r)\diff y  \\
	& =  \sum_{(l,r)\in B_2\cap \{(0,0),(0,1),(1,0)\}} \int_{B_1\cap [0,s]}n\alpha_{\theta_0}\, f_{\theta_0}(y,l,r)\diff y
\end{align*}
follows and therefore $h_{\theta_0} := n\alpha_{\theta_0} f_{\theta_0}$ is the required function. 

The last arguments will show that, under the Assumptions \ref{A1:Fiete}-\ref{A4:Fiete}, the approximate likelihood is indeed \eqref{likehood}.

Note first that $N_n^*$ and $N_0$ are Poisson processes in the same measureable space $(S, \mathcal B)$. Their intensity measure $\nu_n^*$ and $\nu_0$ are finite because  
\begin{align*}
	\nu_n^* (S) & = n \cdot \p_{\theta_0}\big (( Y, L, R) \in S \cap D \big ) = n \alpha_{\theta_0} < \infty \\
	\nu_0 (S) & = \nu_0 \big ( (0, s] \times \{(0,0)\}\big ) + \nu_0 \big([0,s] \times \{(0,1)\}\big) \\
	& \qquad + \nu_0 \big([-(G-s), s] \times \{(1,0)\}\big) + \nu_0 \big((s, \infty)\times \{(1,1)\}\big) \\
	& = \lambda_{[0,s]}\big ((0,s] \big )\cdot \mu_0\big (\{(0,0)\}\big)  +  \lambda_{[0,s]}\big ([0,s]\big )\cdot \mu_0\big (\{(0,1)\}\big) \\
	& \qquad +  \lambda_{[0,s]}\big ([-(G-s), s]\big )\cdot \mu_0\big (\{(1,0)\}\big)  + \lambda_{[0,s]}\big ((s, \infty)\big )\cdot \mu_0\big (\{(1,1)\}\big) \\
	& = s\cdot 1+ s\cdot 1 + s\cdot 1 + 0 = 3s < \infty.
\end{align*}
Finally it is $h_{\theta_0} = n\alpha_{\theta_0} f_{\theta_0}$ a density of $\nu_{n,D} = \nu_n^*$ with respect to $\nu_0$ and by \citet[][Theorem 3.1.1]{reiss1993} a density of $\L(N_n^*)$ with respect to $\L(N_0)$ given by  
\begin{equation*}
	g(\mu)  = \exp\big(\nu_0(S)- \nu_n^*(S)\big) \prod_{l=1}^{\mu(S)} h_{\theta_0}( y^*_l,  l^*_l,  r^*_l) \\
	= \exp (3s - n\alpha_{\theta_0}) \prod_{l=1}^{\mu(S)} n \alpha_{\theta_0} f_{\theta_0}( y^*_l,  l^*_l,  r^*_l)
\end{equation*}
for $\mu$ of the form $\mu = \sum_{l=1}^{\mu(S)} \eps_{( y^*_l,  l^*_l,  r^*_l)}$. Equation \eqref{likehood} results by replacing the true parameter $\theta_0$ by the generic $\theta\in \Theta$ and the expression being evaluated at the observation $n_n^*$ (of $N_n^*$), \eqref{ylrdichte} supplies the form of $f_{\theta_0}$.
\qed

%%%%%%%%%%%%%%%%%%%%%%%%%%%%%%%%%%%%%%%%%%%%%%%%%%%%%%%%%%%%%%%%%%%%%%%%%%%%%%%%%%%%%%%%%%%%%%%%%%%%%%%%%%%%%%%%%

\subsection{Derivation of Properties a-e}

Define $K(\theta):=\log  (1-e^{-\theta (G-s)})$.
\begin{enumerate}
	\item[$(a)$] For $\Theta = [\eps, 1/\eps]$ and all $(y,l,r)\in S$ write 
	$m_\theta(y,l,r)  = \chi_D(y,l,r) \linebreak \{- \log \alpha_{\theta} -\theta  y +K(\theta)  l + \log(\theta) (1- l)(1-r)  \}$.
	Outside $D$, $\theta \mapsto m_\theta(y,l,r)\equiv 0$ and we may restrict to $(y,l,r)\in D$. Especially $y\in [0,s]$ and $l,r\in \{0,1\}$ are bounded on $D$ and $m_\theta$ is a linear combination of $\theta$-dependent functions. In view of \eqref{alpha}, $\alpha_\theta$ is continuous in $\theta$ and bounded away from zero on 
	$\Theta$, so that $\theta \mapsto \log \alpha_\theta$ remains continuous. With an equal argument $\theta \mapsto K(\theta)=\log ( 1- e^{-\theta (G-s)}  )$ and $\theta \mapsto \log \theta$ are continuous functions in  $\theta$ and finally $m_\theta(y,l,r)$ as combination thereof as well.
	\item[$(b)$] By the triangular inequality $\vert m_\theta(Y, L, R) \vert \leq \chi_D(Y, L, R)  \{ \vert \log \alpha_\theta \vert + \theta \vert Y \vert  + \vert K(\theta) \vert \vert L \vert  + \vert \log \theta \vert  \vert (1-L)(1-R) \vert \} 
	\leq \chi_D(Y, L, R) \{  \vert \log \alpha_\theta \vert + \theta s + \vert K(\theta) \vert  + \vert \log \theta \vert \}$, where the second inequality results from $(Y, L, R) \in D$. With a similar line of reasoning as for property $(a)$ the term in brackets depends continuously on $\theta$ and remains bounded on the compact interval $\Theta = [\eps, 1/\eps]$ by $C_0 < \infty$. Therefore $\chi_D(y,l,r) C_0$ is an integrable majorant, because it is $\E_{\theta_0} \chi_D (Y, L, R)  C_0 = \alpha_{\theta_0} C_0 < \infty$.
	
	\item[$(c)$] As $M_n(\theta)$ is an average of the $m_\theta(y,l,r)$, evaluated at the $(Y_i, L_i, R_i)$  ($ i=1,\dots,n$), $(c)$ is true if class $\{m_\theta : \theta \in \Theta\}$ is $P$-Glivenko-Cantelli \cite[see][Theorem 19.4]{vaart1998}. Because $\Theta = [\eps, 1/\eps]$ is compact, \citet[][Problem 19.8]{vaart1998} yields  the property  for $m_\theta$ by properties $(a)$ and $(b)$. Roughly, $(a)$ and $(b)$ do  induce a $L_1$-convergence of the $m_\theta$ and the compactness of $\Theta$ guarantees that the bracketing number of $\{m_\theta : \theta \in \Theta\}$, for given $\eps >0$, remains finite.
	
	\item[$(d)$] The continuity of $\theta \mapsto M(\theta)$ follows from similar reasons as the continuity of $\theta \mapsto m_\theta(y,l,r)$ in property $(a)$. Denote ${\E}^{obs}$ as expectation with respect to ${\p}^{obs}$. 
	$M(\theta)  = \E_{\theta_0} m_\theta(Y, L, R) 
	=  \E_{\theta_0} \chi_D(Y, L, R)  \{ - \log \alpha_{\theta}  -\theta  Y +K(\theta) L + \log(\theta)  (1- L)(1-R)  \}  = \alpha_{\theta_0}  \E_{\theta_0}  [  - \log \alpha_{\theta}  -\theta  Y +K(\theta)  L + \log(\theta)  (1- L)(1-R)   \vert  (Y, L, R) \in D  ]  = \alpha_{\theta_0}  {\E}^{obs}_{\theta_0}  [  - \log \alpha_{\theta}  -\theta  Y^{obs} +K(\theta)  L^{obs} + \log(\theta) (1-  L^{obs})(1- R^{obs}) ]  =  \alpha_{\theta_0} \{  - \log \alpha_{\theta}  -\theta {\E}^{obs}_{\theta_0} Y^{obs} +K(\theta) {\E}^{obs}_{\theta_0}  L^{obs} + \log(\theta) {\E}^{obs}_{\theta_0} (1-  L^{obs})(1- R^{obs}) \}$ (where ${\E}^{obs}$ is the expectation with respect to ${\p}^{obs}$). With $Y^{obs}\in [0,s]$ and $L^{obs},R^{obs}\in \{0,1\}$ all expectations are finite. To be more precise, all constants are positive and $M(\theta)$ is again a linear combination of functions which are continuous in $\theta$, so that $\theta \mapsto M(\theta)$ remains continuous.
	
	\item[$(e)$] Start with the following inequality (that uses the above inequality and the linearity of the expectation):
	\begin{multline}\label{klback} 
		\E_{\theta_0} \chi_D \log \left(\frac{f_{\theta_0}}{f_{\theta}}(Y, L, R) \right)  =  \alpha_{\theta_0} \E_{\theta_0}\left [ \log \left (\frac{f_{\theta_0}}{f_{\theta}}(Y, L, R) \right ) \right. \\
		\left. \vert  (Y, L, R) \in D \right ]  \\
		=  \alpha_{\theta_0} {\E}^{obs}_{\theta_0} \left [ \log \left (\frac{f_{\theta_0}}{f_{\theta}}(Y^{obs}, L^{obs}, R^{obs}) \right ) \right ] \ge 0 
	\end{multline} 	
	
	Without $\alpha_{\theta_0}>0$, the expression \eqref{klback} is the Kulback-Leibler distance between $f_{\theta_0}$ and $f_{\theta}$ and by the basic property of a loss, it is minimal (and zero) if and only if $f_{\theta_0} \equiv f_{\theta}$. By the identification given in Lemma \ref{lemindettrunc}, the latter is true if an only if $\theta_0=\theta$.	
	
	Using the first two equalities yields that $\theta_0$ is also the unique minimizer of $\theta  \mapsto \E_{\theta_0} \chi_D\cdot ( \log f_{\theta_0}  - \log f_{\theta} )= \E_{\theta_0} (m_{\theta_0}- m_\theta) $. Because $\E_{\theta_0} (m_{\theta_0})$ is constant with respect to $\theta$,  $\theta_0$ hence maximizes $M(\theta)= \E_{\theta_0}m_\theta$ uniquely.
	
\end{enumerate}

\qed

%\newpage 

%%%%%%%%%%%%%%%%%%%%%%%%%%%%%%%%%%%%%%%%%%%%%%%%%%%%%%%%%%%%%%%%%%%%%%%%%%%%%%%%%%%%%%%%%%%%%%%%%%%%%%%%%%%%%%%%%

\subsection{Proof of Corollaries \ref{korro1} and \ref{korro2}}

It is $\partial_\theta^2 m_\theta(y,l,r)$ bounded by a integrable majorant $\ddot{m}(y,l,r)$ by the following argument. Define $K''(\theta):= -(G-s)^2(e^{\theta (G-s)}-1)^{-1}  -(G-s)^2(e^{\theta (G-s)}-1)^{-2}$, then
it is $\partial_\theta^2 m_\theta(y,l,r) = \chi_D(y,l,r) \{ -\ddot{\alpha}_\theta \alpha_\theta^{-1} +  \dot{\alpha}_\theta^2 \alpha_\theta^{-2} + K''(\theta) l - \theta^{-2} (1-l)(1-r) \}$. Hence:
\begin{align*}
	|\partial_\theta^2 m_\theta(y,l,r)|& \leq \chi_D(y,l,r) \left \{ \frac{ \vert \ddot{\alpha}_\theta \vert}{\alpha_\theta} + \frac{ \vert \dot{\alpha}_\theta \vert^2 }{\alpha_\theta^2}  + \vert K''(\theta) \vert \vert l \vert + \frac{1}{\theta^2}  \vert(1-l)(1-r)\vert\right \} \\
	& \leq \chi_D(y,l,r)  \left \{\frac{ \vert \ddot{\alpha}_\theta \vert }{\alpha_\theta} + \frac{ \vert \dot{\alpha}_\theta \vert^2 }{\alpha_\theta^2}  + \vert K''(\theta) \vert + \frac{1}{\theta^2} \right \}
\end{align*}
With the same arguments as in the proof for \eqref{lemma9} (see Section \ref{b6}), the denominators in the expressions in brackets are positive (recall $G > s$), the expression itself is continuous, and especially the last expression is bounded over $\Theta=[\eps, 1/\eps]$ by a constant $C<\infty$. Hence it is  $\ddot{m}(y,l,r):=\chi_D(y,l,r)\cdot C$ the integrable majorant.

Both corollaries follow directly. 
\qed 
%\newpage 

%%%%%%%%%%%%%%%%%%%%%%%%%%%%%%%%%%%%%%%%%%%%%%%%%%%%%%%%%%%%%%%%%%%%%%%%%%%%%%%%%%%%%%%%%%%%%%%%%%%%%%%%%%%%%%%%%

\subsection{Proof of Inequality \ref{lemma9}} \label{b6}

Note first that $D^1_\theta m_\theta(y,l,r) = \chi_D(y,l,r)  \{ -\dot{\alpha}_\theta \alpha_\theta^{-1} - y + K'(\theta) l + \theta^{-1} (1-l)(1-r) \}$ with $K'(\theta) := (G-s) / (e^{\theta(G-s)}-1)$ is continuous on $\Theta=[\eps, 1/\eps]$, because all denominators are bounded away from zero and $\dot{\alpha}_\theta$ continuously depends on $\theta$ (see \eqref{alpha}). Furthermore holds
\begin{align*}
	\vert D^1_\theta m_\theta(y,l,r) \vert& \leq \chi_D(y,l,r)  \left \{ \frac{ \vert \dot{\alpha}_\theta \vert }{\alpha_\theta} + \vert y \vert + K'(\theta) \vert l \vert  + \frac{1}{\theta}  \vert (1-l)(1-r)\vert \right \} \\
	& \leq \chi_D(y,l,r) \left \{ \frac{ \vert \dot{\alpha}_\theta \vert }{\alpha_\theta} + s + K'(\theta) + \frac{1}{\theta} \right \},
\end{align*}
where the second inequality results from bounds by $D$ for $y$, $l$ and $r$. The expression in brackets again is continuous in $\theta$, hence it attains a maximum $C < \infty$ on the compact interval $\Theta$. Now it results from the mean value theorem with $\dot{m}(y,l,r):= \chi_D(y,l,r) C$ for a $\overline \theta \in (\theta_1, \theta_2)$
\[\left | \frac{m_{\theta_1}(y,l,r) - m_{\theta_2}(y,l,r) }{\theta_1 - \theta_2}\right | = |D^1_\theta m_{\overline\theta}(y,l,r)| \leq \dot{m}(y,l,r) \]
and hence the inequality. The measurability of $\dot{m}(y,l,r)= \chi_D(y,l,r) C$ is obvious and also $\E_{\theta_0} \dot{m}(Y, L, R)^2 = \E_{\theta_0}\chi_D(Y, L, R) C^2 =   \alpha_{\theta_0} C^2 < \infty$. 

\qed

%\newpage 

%%%%%%%%%%%%%%%%%%%%%%%%%%%%%%%%%%%%%%%%%%%%%%%%%%%%%%%%%%%%%%%%%%%%%%%%%%%%%%%%%%%%%%%%%%%%%%%%%%%%%%%%%%%%%%%%%

\subsection{Proof of Formulae \ref{elelr}} 

The expectations use the definitions of $L$ and $R$ from Section \ref{mal}.
\begin{align*}
	\E_{\theta} \chi_D(Y, L, R) L  & = \alpha_{\theta} \E_{\theta} [L  \vert  (Y, L, R) \in D] 
	=  \alpha_{\theta}  \E_{\theta}  [\chi_{\{T > 0\}} \vert  (X,T) \in D_0 ] \\
	& =  \alpha_{\theta}   (1- \p_{\theta} \big[T \leq 0  \vert  (X,T) \in D_0 ] )
\end{align*}
The event $\{T \leq 0\}$ now implies, due to the shape of $D_0$ and according to Assumption \ref{A4:Fiete}, already  $\{(X,T) \in D_0\}$ and due to the uniform distribution of $T$ over $[-s, G-s]$ (see Assumption \ref{A2:Fiete}) follows	$\E_{\theta} \chi_D(Y, L, R)  L   =\alpha_{\theta}  (1- \p_{\theta}[T \leq 0 ] / \alpha_{\theta} ) = \alpha_{\theta} - s/G$.
Analogous statements hold for the second expectation and with the CDF $F^{\text{Exp}}(x \vert \theta)= 1-e^{-\theta x}$ of $X$ (see again Assumption \ref{A2:Fiete}) it is 
\begin{multline*}
	\E_{\theta} \chi_D(Y, L, R) ( 1-L)(1-R) = \alpha_{\theta}  \E_{\theta}   [\chi_{\{T \leq 0\}} \chi_{\{T+s \geq X\}}  \vert (X,T) \in D_0 ] \\
	=  \alpha_{\theta}  \p_{\theta}[T \leq 0, \, T +s \geq X] / \alpha_{\theta} 
	= \frac{1}{G} \int_{-s}^0  \{ 1- e^{-\theta (u+s)} \} \diff u 
	= \frac{s}{G} - \frac{1}{G \theta}  (1-e^{-\theta s} ).
\end{multline*}

\qed

%\newpage

%%%%%%%%%%%%%%%%%%%%%%%%%%%%%%%%%%%%%%%%%%%%%%%%%%%%%%%%%%%%%%%%%%%%%%%%%%%%%%%%%%%%%%%%%%%%%%%%%%%%%%%%%%%%%%%%%

\subsection{Proof of Inequality \ref{lemdelta}}

First of all note that $\Delta= \alpha_\theta H_2(\theta) - \frac{1}{G} (1-e^{-\theta s}  ) H_3(\theta)$. Define 
$a := s/(1-e^{-\theta s})$, $b:= \theta^{-1}$ and $c:= e^{-\theta s}$. Then $H_k(\theta)= a^k \cdot c -b^k$ (for $k=1,2,3$) and it is $H_1(\theta)(a+b) = (ac -b)(a+b) = a^2c-b^2 + ab(c-1) = H_2(\theta) - ab(1-e^{-\theta s})$ so that $H_2(\theta)  = H_1(\theta) (a+b) +  ab (1-e^{-\theta s})$ and $H_2(\theta)(a+b)  = (a^2c-b^2)(a+b)= a^3c - b^3 + ab(ac-b) = H_3(\theta) + ab H_1(\theta)$
\begin{align*} 
	\Rightarrow \quad H_3(\theta)&= H_2(\theta)(a+b) - ab H_1(\theta) \\
&	= H_1(\theta) (a+b)^2 + ab(a+b) (1-e^{-\theta s}) - ab H_1(\theta) \\
	& = H_1(\theta) (a^2 + ab + b^2)+ ab(a+b) (1-e^{-\theta s}).
\end{align*}
Furthermore we have
\begin{equation*}
	\alpha_\theta   =  \frac{s}{G} + \frac{s}{G} \frac{1}{ \theta} \frac{(1-e^{-\theta s})}{s}(1-e^{-\theta (G- s)}) 
	=\frac{s}{G}  \left ( 1+b \cdot \frac{1}{a}  (1-e^{-\theta (G- s)}) \right )
\end{equation*}
and $ (1-e^{-\theta s})/G  =s/(Ga)$.
Therefore 
\begin{align*}
	\alpha_\theta H_2(\theta) & = \frac{s}{G}  ( 1+b  a^{-1} \cdot (1-e^{-\theta (G- s)}))   ( H_1(\theta)\cdot (a+b) +  ab  (1-e^{-\theta s})  ) \\
	& = \frac{s}{G}  ( H_1(\theta) (a+b) +  ab  (1-e^{-\theta s})  ) \\
	& \qquad \qquad \qquad+ \frac{s}{G}   (H_1(\theta) (b+b^2/a) +  b^2  (1-e^{-\theta s})  ) (1-e^{-\theta (G-s)})\\
	& = \frac{s}{G}   ( H_1(\theta)  (a+2b+b^2/a) +  (ab+b^2)  (1-e^{-\theta s}) ) \\
	& \qquad \qquad \qquad- \frac{s}{G}   (H_1(\theta) (b+b^2/a) +  b^2   (1-e^{-\theta s}) ) e^{-\theta (G-s)}
\end{align*}
and 
\begin{align*}
	\frac{1}{G} (1-e^{-\theta s}  ) H_3(\theta) & =  \frac{s}{Ga}     ( H_1(\theta) (a^2 + ab + b^2)+ ab(a+b)  (1-e^{-\theta s}) ) \\
	& =  \frac{s}{G}   ( H_1(\theta) (a + b + b^2/a)+ (ab+b^2) (1-e^{-\theta s})).
\end{align*}
The difference of the last two expressions is
\begin{align*}
	\Delta & = \alpha_\theta H_2(\theta) - \frac{1}{G} (1-e^{-\theta s}  ) H_3(\theta) \\
	& =  \frac{s}{G}   ( H_1(\theta)  b )- \frac{s}{G}   (H_1(\theta)(b+b^2/a) +  b^2  (1-e^{-\theta s}) ) e^{-\theta (G-s)} \\
	& =  \frac{s}{G}  ( H_1(\theta)  b )(1-e^{-\theta (G-s)})-  \frac{s}{G}   (H_1(\theta) b^2/a +  b^2   (1-e^{-\theta s})  ) e^{-\theta (G-s)} \\
	& =  \frac{s}{G}  b  H_1(\theta)(1-e^{-\theta (G-s)})-  \frac{s}{G}  b^2   (H_1(\theta)/a +  (1-e^{-\theta s}) ) e^{-\theta (G-s)} \\
	& =  \frac{s}{G\theta} H_1(\theta)(1-e^{-\theta (G-s)})-  \frac{s}{G\theta^2}   (H_1(\theta)/s +  1    ) (1-e^{-\theta s}) e^{-\theta (G-s)},
\end{align*}
where $b= 1/ \theta$ und $1/a = (1-e^{-\theta s})/s$ enter in the last equation. This is the expression $\Delta$ on the left side of inequality \eqref{lemdelta}. In order to derive the negativity, note that due to $\Theta = [\eps, 1/ \eps]$ for $\eps >0$ and $0<s< G$ all factors that do not include $H_1(\theta)$ are strictly positive. We show now $H_1(\theta) < 0$ and $H_1(\theta)/s + 1 > 0$ for all $\theta \in \Theta$. It is known that $e^x \geq x+1$ for all $x\in \R$ with equality only in $x=0$ so that due to $s>0$ and $\theta s>0$:
\[e^{\theta s} > \theta s + 1 \quad \Leftrightarrow \quad  \frac{1}{e^{\theta s}- 1} < \frac{1}{\theta s} \quad \Leftrightarrow \quad H_1(\theta) = \frac{s e^{-\theta s}}{1-e^{-\theta s}} - \frac{1}{\theta} < 0.\]
For $H_1(\theta)/s + 1 > 0$ consider  
\begin{multline*}  \frac{1}{s} H_1(\theta) + 1  =  \frac{ e^{-\theta s}}{1-e^{-\theta s}} - \frac{1}{\theta s} + 1  = \frac{ 1 + e^{\theta s} - 1}{e^{\theta s}-1} - \frac{1 }{\theta s} \\  = \frac{\theta s e^{\theta s} - e^{\theta s} + 1}{(e^{\theta s}-1)\theta s} =  \frac{(\theta s - 1) e^{\theta s} + 1}{(e^{\theta s}-1)\theta s}. 
\end{multline*}
The denominator of the ratio on the right hand side is positive over $\Theta$. The numerator has a root in $\theta = 0$ and is strictly increasing, so that it is only positive over $\Theta$ and this inequality is true. Finally follows as stated $\Delta < 0$.

\qed


\begin{thebibliography}{28}
	\expandafter\ifx\csname natexlab\endcsname\relax\def\natexlab#1{#1}\fi
	\providecommand{\url}[1]{\texttt{#1}}
	\providecommand{\href}[2]{#2}
	\providecommand{\path}[1]{#1}
	\providecommand{\DOIprefix}{doi:}
	\providecommand{\ArXivprefix}{arXiv:}
	\providecommand{\URLprefix}{URL: }
	\providecommand{\Pubmedprefix}{pmid:}
	\providecommand{\doi}[1]{\href{http://dx.doi.org/#1}{\path{#1}}}
	\providecommand{\Pubmed}[1]{\href{pmid:#1}{\path{#1}}}
	\providecommand{\bibinfo}[2]{#2}
	\ifx\xfnm\relax \def\xfnm[#1]{\unskip,\space#1}\fi
	%Type = Article
	\bibitem[{Alan et~al.(2012)Alan, Honoré, Hu and Leth-Petersen}]{honore2012}
	\bibinfo{author}{Alan, S.}, \bibinfo{author}{Honoré, B.}, \bibinfo{author}{Hu,
		L.}, \bibinfo{author}{Leth-Petersen, S.}, \bibinfo{year}{2012}.
	\newblock \bibinfo{title}{Estimation of panel data regression models with
		two-sided censoring or truncation}.
	\newblock \bibinfo{journal}{Journal of Econometric Methods}
	\bibinfo{volume}{3}, \bibinfo{pages}{1--20}.
	%Type = Incollection
	\bibitem[{Andersen et~al.(1988)Andersen, Borgan, Gill and Keiding}]{And0}
	\bibinfo{author}{Andersen, P.}, \bibinfo{author}{Borgan, {\O}.},
	\bibinfo{author}{Gill, R.}, \bibinfo{author}{Keiding, N.},
	\bibinfo{year}{1988}.
	\newblock \bibinfo{title}{Censoring, truncation and filtering in statistical
		models based on counting processes}, in: \bibinfo{editor}{Prabhu, N.U.}
	(Ed.), \bibinfo{booktitle}{Statistical inference from stochastic processes}.
	\bibinfo{publisher}{Center for Mathematics and Computer Science, Amsterdam}.
	volume~\bibinfo{volume}{80}, pp. \bibinfo{pages}{19--60}.
	%Type = Article
	\bibitem[{Efron and Petrosian(1999)}]{efron1999}
	\bibinfo{author}{Efron, B.}, \bibinfo{author}{Petrosian, V.},
	\bibinfo{year}{1999}.
	\newblock \bibinfo{title}{Nonparametric methods for doubly truncated data}.
	\newblock \bibinfo{journal}{Journal of the American Statistical Association}
	\bibinfo{volume}{94}, \bibinfo{pages}{824--834}.
	%Type = Article
	\bibitem[{Emura et~al.(2017)Emura, Hu and Konno}]{emura2017}
	\bibinfo{author}{Emura, T.}, \bibinfo{author}{Hu, Y.H.},
	\bibinfo{author}{Konno, Y.}, \bibinfo{year}{2017}.
	\newblock \bibinfo{title}{Asymptotic inference for maximum likelihood
		estimators under the special exponential family with double-truncation}.
	\newblock \bibinfo{journal}{Statistical Papers} \bibinfo{volume}{58},
	\bibinfo{pages}{877--909}.
	%Type = Article
	\bibitem[{Emura and Pan(2020)}]{emura2020}
	\bibinfo{author}{Emura, T.}, \bibinfo{author}{Pan, C.H.}, \bibinfo{year}{2020}.
	\newblock \bibinfo{title}{Parametric likelihood inference and goodness-of-fit
		for dependently left-truncated data, a copula approach}.
	\newblock \bibinfo{journal}{Statistical Papers} \bibinfo{volume}{61},
	\bibinfo{pages}{479--501}.
	%Type = Article
	\bibitem[{Eriksson et~al.(2015)Eriksson, Martinussen and
		Scheike}]{eriksson2015}
	\bibinfo{author}{Eriksson, F.}, \bibinfo{author}{Martinussen, T.},
	\bibinfo{author}{Scheike, T.}, \bibinfo{year}{2015}.
	\newblock \bibinfo{title}{Clustered survival data with left-truncation}.
	\newblock \bibinfo{journal}{Scandinavian Journal of Statistics}
	\bibinfo{volume}{42}, \bibinfo{pages}{1149--1166}.
	%Type = Article
	\bibitem[{Frank et~al.(2019)Frank, Chae and Kim}]{franchae2019}
	\bibinfo{author}{Frank, G.}, \bibinfo{author}{Chae, M.}, \bibinfo{author}{Kim,
		Y.}, \bibinfo{year}{2019}.
	\newblock \bibinfo{title}{Additive time-dependent hazard model with doubly
		truncated data}.
	\newblock \bibinfo{journal}{Journal of the Korean Statistical Society}
	\bibinfo{volume}{48}, \bibinfo{pages}{179--193}.
	%Type = Article
	\bibitem[{Gross and Lai(1996)}]{gross1996}
	\bibinfo{author}{Gross, S.}, \bibinfo{author}{Lai, T.}, \bibinfo{year}{1996}.
	\newblock \bibinfo{title}{Nonparametric estimation and regression analysis with
		left-truncated and right-censored data}.
	\newblock \bibinfo{journal}{Journal of the American Statistical Association}
	\bibinfo{volume}{91}, \bibinfo{pages}{1166--1180}.
	%Type = Article
	\bibitem[{He and Yang(2003)}]{he2003}
	\bibinfo{author}{He, S.}, \bibinfo{author}{Yang, G.}, \bibinfo{year}{2003}.
	\newblock \bibinfo{title}{Estimation of regression parameters with left
		truncated data}.
	\newblock \bibinfo{journal}{Journal of Statistical Planning and Inference}
	\bibinfo{volume}{117}, \bibinfo{pages}{99--122}.
	%Type = Article
	\bibitem[{Heckman(1976)}]{heckman1976}
	\bibinfo{author}{Heckman, J.}, \bibinfo{year}{1976}.
	\newblock \bibinfo{title}{The common structure of statistical models of
		truncation, sample selection and limited dependent variables and a simple
		estimator for such models}.
	\newblock \bibinfo{journal}{Annals of Economic and Social Measurement}
	\bibinfo{volume}{5}, \bibinfo{pages}{475--492}.
	%Type = Article
	\bibitem[{Kalbfleisch and Lawless(1989)}]{kalblawl1989}
	\bibinfo{author}{Kalbfleisch, J.}, \bibinfo{author}{Lawless, J.},
	\bibinfo{year}{1989}.
	\newblock \bibinfo{title}{Inference based on retrospective ascertainment: An
		analysis of the data on transfusion-related {AIDS}}.
	\newblock \bibinfo{journal}{Journal of the American Statistical Association}
	\bibinfo{volume}{84}, \bibinfo{pages}{360--372}.
	%Type = Article
	\bibitem[{Keiding and Gill(1990)}]{keiding1990}
	\bibinfo{author}{Keiding, N.}, \bibinfo{author}{Gill, R.D.},
	\bibinfo{year}{1990}.
	\newblock \bibinfo{title}{Random truncation models and {M}arkov processes}.
	\newblock \bibinfo{journal}{Annals of Statistics} \bibinfo{volume}{18},
	\bibinfo{pages}{582--602}.
	%Type = Article
	\bibitem[{Lai and Ying(1991)}]{lai1991}
	\bibinfo{author}{Lai, T.L.}, \bibinfo{author}{Ying, Z.}, \bibinfo{year}{1991}.
	\newblock \bibinfo{title}{Estimating a distribution function with truncated and
		censored data}.
	\newblock \bibinfo{journal}{Annals of Statistics} \bibinfo{volume}{19},
	\bibinfo{pages}{417--442}.
	%Type = Article
	\bibitem[{Lynden-Bell(1971)}]{lynden1971}
	\bibinfo{author}{Lynden-Bell, D.}, \bibinfo{year}{1971}.
	\newblock \bibinfo{title}{A method of allowing for known observational
		selection in small samples applied to {3CR} quasars}.
	\newblock \bibinfo{journal}{Monthly Notices of the Royal Astronomical Society}
	\bibinfo{volume}{155}, \bibinfo{pages}{95--118}.
	%Type = Article
	\bibitem[{Moreira and de~U{\~n}a-{\'A}lvarez(2010)}]{moreira2010b}
	\bibinfo{author}{Moreira, C.}, \bibinfo{author}{de~U{\~n}a-{\'A}lvarez, J.},
	\bibinfo{year}{2010}.
	\newblock \bibinfo{title}{Bootstrapping the {NPMLE} for doubly truncated data}.
	\newblock \bibinfo{journal}{Journal of Nonparametric Statistics}
	\bibinfo{volume}{22}, \bibinfo{pages}{567--583}.
	%Type = Article
	\bibitem[{Moreira et~al.(2016)Moreira, U{\~n}a-{\'A}lvarez and
		Meira-Machado}]{moreira2016}
	\bibinfo{author}{Moreira, C.}, \bibinfo{author}{U{\~n}a-{\'A}lvarez, J.},
	\bibinfo{author}{Meira-Machado, L.}, \bibinfo{year}{2016}.
	\newblock \bibinfo{title}{Nonparametric regression with doubly truncated data}.
	\newblock \bibinfo{journal}{Computational Statistics and Data Analysis}
	\bibinfo{volume}{93}, \bibinfo{pages}{294--307}.
	%Type = Book
	\bibitem[{Reiss(1993)}]{reiss1993}
	\bibinfo{author}{Reiss, R.D.}, \bibinfo{year}{1993}.
	\newblock \bibinfo{title}{A Course on Point Processes}.
	\newblock \bibinfo{publisher}{Springer, New York}.
	%Type = Misc
	\bibitem[{Scholz and {Wei\ss bach}(2024)}]{scholzweis2025}
	\bibinfo{author}{Scholz, E.}, \bibinfo{author}{{Wei\ss bach}, R.},
	\bibinfo{year}{2024}.
	\newblock \bibinfo{title}{Left-truncated discrete lifespans: The {AFiD}
		enterprise panel}.
	\newblock \URLprefix \url{https://arxiv.org/abs/2411.12367}.
	%Type = Article
	\bibitem[{Shen(2010)}]{shen2010}
	\bibinfo{author}{Shen, P.S.}, \bibinfo{year}{2010}.
	\newblock \bibinfo{title}{Nonparametric analysis of doubly truncated data}.
	\newblock \bibinfo{journal}{Annals of the Institute of Statistical Mathematics}
	\bibinfo{volume}{62}, \bibinfo{pages}{835--853}.
	%Type = Article
	\bibitem[{Shen(2014)}]{shen2014a}
	\bibinfo{author}{Shen, P.S.}, \bibinfo{year}{2014}.
	\newblock \bibinfo{title}{Nonparametric estimation with left-truncated and
		right-censored data when the sample size before truncation is known}.
	\newblock \bibinfo{journal}{Statistics} \bibinfo{volume}{48},
	\bibinfo{pages}{315--326}.
	%Type = Article
	\bibitem[{Stute(1993)}]{stute1993b}
	\bibinfo{author}{Stute, W.}, \bibinfo{year}{1993}.
	\newblock \bibinfo{title}{Almost sure representations of the product-limit
		estimator for truncated data}.
	\newblock \bibinfo{journal}{Annals of Statistics} \bibinfo{volume}{21},
	\bibinfo{pages}{146--156}.
	%Type = Article
	\bibitem[{Toparkus and {Wei\ss bach}(2025)}]{topaweis2024}
	\bibinfo{author}{Toparkus, A.M.}, \bibinfo{author}{{Wei\ss bach}, R.},
	\bibinfo{year}{2025}.
	\newblock \bibinfo{title}{Testing truncation dependence: The {G}umbel-{B}arnett
		copula}.
	\newblock \bibinfo{journal}{Journal of Statistical Planning and Inference}
	\bibinfo{volume}{234}, \bibinfo{pages}{106194}.
	%Type = Article
	\bibitem[{Turnbull(1976)}]{turnbull1976}
	\bibinfo{author}{Turnbull, B.W.}, \bibinfo{year}{1976}.
	\newblock \bibinfo{title}{The empirical distribution function with arbitrarily
		grouped, censored and truncated data}.
	\newblock \bibinfo{journal}{Journal of the Royal Statistical Society. Series B
		(Methodological)} \bibinfo{volume}{38}, \bibinfo{pages}{290--295}.
	%Type = Book
	\bibitem[{van~der Vaart(1998)}]{vaart1998}
	\bibinfo{author}{van~der Vaart, A.}, \bibinfo{year}{1998}.
	\newblock \bibinfo{title}{Asymptotic Statistics}.
	\newblock \bibinfo{publisher}{Cambridge University Press, Cambridge}.
	%Type = Article
	\bibitem[{Weißbach et~al.(2024)Weißbach, Dörre, Wied, Doblhammer and
		Fink}]{weissbachm2021effect}
	\bibinfo{author}{Weißbach, R.}, \bibinfo{author}{Dörre, A.},
	\bibinfo{author}{Wied, D.}, \bibinfo{author}{Doblhammer, G.},
	\bibinfo{author}{Fink, A.}, \bibinfo{year}{2024}.
	\newblock \bibinfo{title}{Left-truncated health insurance claims data:
		Theoretical review and empirical application}.
	\newblock \bibinfo{journal}{AStA Advances in Statistical Analysis}
	\bibinfo{volume}{108}, \bibinfo{pages}{31--68}.
	%Type = Article
	\bibitem[{{Wei\ss bach} and {D\"orre}(2022)}]{weisdoer2022}
	\bibinfo{author}{{Wei\ss bach}, R.}, \bibinfo{author}{{D\"orre}, A.},
	\bibinfo{year}{2022}.
	\newblock \bibinfo{title}{Retrospective sampling of survival data based on a
		poisson birth process: Conditional maximum likelihood}.
	\newblock \bibinfo{journal}{Statistics} \bibinfo{volume}{56},
	\bibinfo{pages}{844--866}.
	%Type = Article
	\bibitem[{{Wei\ss bach} and Wied(2022)}]{weiswied2021}
	\bibinfo{author}{{Wei\ss bach}, R.}, \bibinfo{author}{Wied, D.},
	\bibinfo{year}{2022}.
	\newblock \bibinfo{title}{Truncating the exponential with a uniform
		distribution}.
	\newblock \bibinfo{journal}{Statistical Papers} \bibinfo{volume}{63},
	\bibinfo{pages}{1247--1270}.
	%Type = Article
	\bibitem[{Woodroofe(1985)}]{woodroofe1985}
	\bibinfo{author}{Woodroofe, M.}, \bibinfo{year}{1985}.
	\newblock \bibinfo{title}{Estimating a distribution function with truncated
		data}.
	\newblock \bibinfo{journal}{Annals of Statistics} \bibinfo{volume}{13},
	\bibinfo{pages}{163--177}.
	
\end{thebibliography}
\end{document}